\documentclass[letterpaper, 10 pt, conference]{ieeeconf}  % Comment this line out if you need a4paper

\IEEEoverridecommandlockouts                              % This command is only needed if 
\usepackage{hyperref}
\usepackage{url}
\usepackage{graphicx}
\usepackage{hyperref}
\usepackage{comment}
  \usepackage{amsthm}
\usepackage{amsmath}
\usepackage{amssymb}
\usepackage{amsfonts}
\usepackage{amsopn}
\usepackage{diagbox}
\usepackage{cite}
\usepackage{hyperref}
\hypersetup{
    colorlinks=true,
    linkcolor=blue,
    filecolor=blue,      
    urlcolor=blue,
    citecolor=cyan,
}
\usepackage[dvipsnames]{xcolor}
\newtheorem{theorem}{Theorem}[]
\newtheorem{proposition}{Proposition}[]
\newtheorem{lemma}{Lemma}
\newtheorem{remark}{Remark}
\newtheorem{problem}{Problem}
\newtheorem{assumption}{Assumption}
\newtheorem{definition}{Definition}[]

\newcommand{\R}{\mathbb{R}}
\newcommand{\ub}{\mathbf{u}}
\newcommand{\x}{\mathbf{x}}
\newcommand{\A}{\mathbf{A}}
\newcommand{\B}{\mathbf{B}}
\newcommand{\M}{\mathbf{M}}
\newcommand{\Sc}{\mathcal{S}}
\newcommand{\Cn}{\mathcal{C}_n}
\newcommand{\Pc}{\mathbb{P}}

\title{\LARGE \bf
On the Hardness of Learning to Stabilize Linear Systems
}

\author{Xiong Zeng$^1$ \quad\;\; Zexiang Liu$^1$ \quad\;\; Zhe Du$^1$ \quad\;\; Necmiye Ozay$^1$ \quad\;\; Mario Sznaier$^2$
 \thanks{This work is supported in part by ONR CLEVR-AI MURI (\#N00014-21-1-2431) and NSF CNS Grant \#1931982.}
\thanks{$^1$ Department of Electrical Engineering and Computer Science, University
of Michigan Ann Arbor, MI 48105. Emails: \{zengxion, zexiang, zhedu, necmiye\}@umich.edu.
       }%
\thanks{$^2$ Department of Electrical and Computer Engineering,  Northeastern University, Boston, MA 02115. Email: msznaier@ece.neu.edu.}
}

\begin{document}

\maketitle
\thispagestyle{empty}
\pagestyle{empty}

%%%%%%%%%%%%%%%%%%%%%%%%%%%%%%%%%%%%%%%%%%%%%%%%%%%%%%%%%%%%%%%%%%%%%%%%%%%%%%%%
\begin{abstract}
Inspired by the work of Tsiamis et al. \cite{tsiamis2022learning}, in this
paper we study the statistical hardness of learning to stabilize
linear time-invariant systems. Hardness is measured by the number of samples required to achieve a learning task with a given probability. The work in \cite{tsiamis2022learning} shows that there
exist system classes that are hard to learn
to stabilize with the core reason being the hardness of identification. Here we present a class of systems that can be
easy to identify, thanks to a non-degenerate noise process that
excites all modes, but the sample complexity of stabilization
still increases exponentially with the system dimension.
We tie this result to the hardness of co-stabilizability for this
class of systems using ideas from robust control. 

\end{abstract}

%%%%%%%%%%%%%%%%%%%%%%%%%%%%%%%%%%%%%%%%%%%%%%%%%%%%%%%%%%%%%%%%%%%%%%%%%%%%%%%%
\section{INTRODUCTION}

Learning-based control plays an increasingly important role in many application domains such as power systems \cite{liu2022stability}, robotics \cite{brunke2022safe}, self-driving cars \cite{schwarting2018planning}, where it might be hard to perfectly model the system and its environment. Many learning-based control algorithms assume the existence of an initial stabilizing controller in order to simplify their analysis. Such simplifying assumptions are prevalent both in model-based \cite{cohen2019learning,simchowitz2020improper,zheng2021sample,ouyang2019posterior,faradonbeh2020optimism,dean2018regret,sattar2021identification} and model-free \cite{fazel2018global,cassel2020logarithmic,zhang2020policy,tang2021analysis,li2021distributed,mohammadi2021convergence,yang2019provably} learning-based control algorithms. However, learning to stabilize is a fundamental problem in learning-based control, with several algorithms tackling this issue\cite{abbasi2011regret,faradonbeh2018finite,lale2022reinforcement,chen2021black,hu2022sample,dai2020data,perdomo2021stabilizing}.  

Understanding the fundamental limits or the corner cases of learning-to-stabilize algorithms can inform future algorithm design and is crucial for applications of these algorithms in safety-critical domains. Therefore, it is important to understand how the system properties affect the performance of the learning-to-stabilize algorithms. In particular, we are interested in the number of samples required to learn a stabilizing controller with a given probability as a performance measure. We say a class of systems is hard to learn to stabilize if this number grows exponentially with the system dimension, independent of the algorithm choice.  

We focus on fully observed linear time-invariant systems and consider the task of learning a static stabilizing linear state-feedback controller from a single trajectory. In this setting, Tsiamis et al. \cite{tsiamis2022learning} show that when the process noise is degenerate, i.e. the noise covariance matrix being singular, there are some classes of systems that are hard to learn to stabilize, by transferring the hardness of learning-to-stabilize into the hardness of system identification. The system classes constructed in their work are based on a (marginally) stable hard-to-stabilize pair. In this work, we significantly extend the class of systems that are hard to learn to stabilize by considering systems that are, even though close in the parameter space and generate similar state-input trajectories, not co-stabilizable with the same controller. This is achieved by a novel analysis technique that uses Ackermann's formula to compute all stabilizing linear state-feedback gains analytically and characterize the minimal level of perturbations to the parameters that render co-stabilizability infeasible. Different from the prior work, our analysis allows us to consider system classes that may only
include systems with eigenvalues strictly outside of the unit circle, for which stabilizability is arguably more critical.

%\subsection{Notations}
\noindent{\bf Notation:} We use lower case, lower case boldface, and upper case boldface letters to denote scalars,
vectors, and matrices respectively. For a matrix $\mathbf{M}\in \mathbb{R}^{m\times n}$, $\mathbf{M}^{\top}$ denotes its transpose,  $M^{(i,j)}$ denotes its element in the $i^{th}$ row and the $j^{th}$ column. For a square matrix $\mathbf{M} \in \mathbb{R}^{n \times n}$, $\mathbf{M}\succ 0$ ($\succeq 0$) denotes that $\mathbf{M}$ is positive definite (positive semidefinite), $\rho(\M)$ denotes its spectral radius, and $\operatorname{det}(\M)$ denotes its determinant. For a vector $\mathbf{v} \in \mathbb{R}^n$, its $i^{th}$ element is denoted by $v^{(i)}$. By $\operatorname{poly}(\cdot)$ we denote a polynomial function of
its arguments. By $\operatorname{exp}(\cdot)$ we denote an exponential function of its arguments. We use $\mathbf{I}_n$ to denote the identity matrix in $\mathbb{R}^{n \times n}$. A sequence of vectors $\x_{t}$, $\x_{t+1}$, ..., $\x_{t+N}$ is denoted by $\x_{t:t+N}$ for short. {By convention, $\x_{i:j}$ is an empty set if $j<i$.}

\section{Problem Setup and Preliminary Notions}\label{sec2}
%\subsection{Discrete-time Linear Time-Invariant System}
We consider the following fully-observed discrete-time linear time-invariant (LTI) system:
\begin{equation}
\mathbf{x}_{t+1} =\mathbf{A} \mathbf{x}_t + \mathbf{B} \mathbf{u}_t +  \mathbf{w}_t,
\label{eq1}
\end{equation}
where $\mathbf{x}_t \in \mathbb{R}^{n}$, $\mathbf{u}_t \in \mathbb{R}^{p}$, $\mathbf{w}_t \in \mathbb{R}^{n}$ are the state, input, and process noise at time $t$. For simplicity, we assume $ \mathbf{x}_{0} = \mathbf{0}$. The random process $\mathbf{w}_t$ over $t$ is zero-mean i.i.d. Gaussian, with covariance matrix $\sigma_w^2 \mathbf{I}_n$. 
In the remainder of the paper, we denote a system in the form \eqref{eq1} by the tuple $(\mathbf{A,B})$.

Let $\mathcal{C}_n$ be a class of systems $(\mathbf{A}, \mathbf{B})$ in dimension $n$, parameterized by some unknown parameters. 

\begin{definition}
A \emph{learning-to-stabilize algorithm} $\pi$ with respect to the class $\mathcal{C}_n$ is a sequence of functions $\pi =\{\pi_t\}_{t=0}^N$. For $t=0,$ ..., $N-1$, $\pi_t(\mathbf{u}_{0:t-1},\mathbf{x}_{0:t})$ specifies the probability distribution of the input $\mathbf{u}_t\in \R^p$ at time $t$, conditioned on the previous state-input trajectory $\mathbf{u}_{0:t-1}$ and $\mathbf{x}_{0:t}$. Then at $t=N$, the function $\pi_N$ maps the entire state-input trajectory $\mathbf{u}_{0:N-1}$ and $\mathbf{x}_{0:N}$ to a state-feedback gain in $\mathbb{R}^{p\times n}$.
This learned state-feedback gain $\hat{\mathbf{K}}_N = \pi_N(\ub_{0:N-1}, \x_{1:N})$ is called \emph{stabilizing} if ${\rho(\mathbf{A}+\mathbf{B} \hat{\mathbf{K}}_N)<1}$.
    \label{def1}
\end{definition}

Intuitively, the algorithm $\pi$ consists of an exploration policy in the first $N-1$ steps and decides on the gain $\hat{\mathbf{K}}_N$ using the data generated during exploration at step $N$. As such, exciting the system with some open-loop persistently exciting input as in data-driven control~\cite{de2019formulas}, applying some i.i.d. input and computing the gain afterward using the generated data~\cite{mania2019certainty}, or active learning policies can all be considered as special types of learning-to-stabilize algorithms.

Given a system $\Sc=(\A, \B) \in \Cn$ and a learning-to-stabilize algorithm $\pi$, let $\Pc_{\Sc,\pi}^N$ denote the probability measure of the input-state samples $\ub_{0:N-1}$ and $\x_{1:N}$ (with $f^N_{\Sc,\pi}$ denoting the corresponding probability density function), and $\mathbb{E}^N_{\mathcal{S,\pi}}$ denote the expectation of the respective probability measure.
We make the following assumptions on the class $\mathcal{C}_{n}$ and the algorithm $\pi$.

\begin{assumption}
    For all $n\geq 1$ and all $(\A,\B) \in \mathcal{C}_n$, the norm of matrices $\mathbf{A}, \mathbf{B}$ is bounded by a positive constant $M$, that is, $\max_{n\geq 1,(\A,\B)\in \mathcal{C}_n} \max\left\{\|\mathbf{A}\|_2,\|\mathbf{B}\|_2 \right\} \leq M$. 
    \label{asm1}
\end{assumption}
\begin{assumption}The second moment of the norm of the input signal $\ub_t$, generated by the algorithm $\pi$, is bounded by some constant $\sigma_u^2>0$. That is, $\mathbb{E}_{\Sc, \pi}\left[ \left\|\mathbf{u}_t\right\|_2^2\right] \leq \sigma_u^2$.
    \label{asm2}
\end{assumption}

Next, we recall the definition of $\operatorname{poly}(n)$-stabilizable system classes from  \cite{tsiamis2022learning}. If a class $\Cn$ of discrete-time LTI systems is $\operatorname{poly}(n)$-stabilizable, it is statistically easy to learn linear state-feedback controllers to stabilize systems in this class.  
\begin{definition}[$\operatorname{poly}(n)$-stabilizable system classes \cite{tsiamis2022learning}]
Under Assumptions~\ref{asm1} and \ref{asm2},  a class $\mathcal{C}_n$ of systems is $\operatorname{poly}(n)$-stabilizable if there exists a learning-to-stabilize algorithm $\pi$ such that for all confidence levels $0 \leq \delta<1$
\begin{align}
\inf_{\mathcal{S} \in \mathcal{C}_n} \mathbb{P}^N_{\mathcal{S},\pi  }\left(\rho\left(\mathbf{A} + \mathbf{B} \pi_N(\ub_{0:N-1}, \x_{1:N})\right) < 1\right) > 1 - \delta,
\end{align}
if the sample size $N$ satisfies
$N \sigma_u^2 \geq \operatorname{poly}(n, \log (1 / \delta), M).$
\label{def2}
\end{definition}

This definition essentially tells that a class is $\operatorname{poly}(n)$-stabilizable if it is possible to find an algorithm that can learn a stabilizing linear state-feedback controller with high probability, even for the worst-case system in this class, as long as there are polynomially many samples in the system dimension $n$. Since the polynomial dependency on $n$ is mild, we say learning to stabilize is \emph{easy} for this class. On the other hand, being \emph{hard} refers to a class that is not $\text{poly}(n)$-stabilizable.

A closely related concept is the hardness of identification \cite{tsiamis2021linear}, i.e., whether the system can be learned with $\epsilon$ accuracy using $\text{poly}(n, \log(1/\delta), 1/\epsilon)$ many samples.
When the process noise is degenerate, by transferring the hardness of learning to stabilize into the hardness of system identification, Tsiamis et al. \cite{tsiamis2022learning} prove that there exists a class of systems, for which the worst-case sample complexity of learning to stabilize is at least exponential with the system dimension. Our work is complementary as we seek to answer the following question.

\begin{problem}
    Is there a class of linear systems that are not $\operatorname{poly}(n)$-stabilizable when the process noise $\mathbf{w}_t$ is non-degenerate? 
    \label{prob1}
\end{problem}

The following lemma follows directly from Definition~\ref{def2}. 

\begin{lemma}
    For two classes of systems $\mathcal{C}^1_n$ and $\mathcal{C}^2_n$, if $\mathcal{C}^1_n$ is a subset of $\mathcal{C}^2_n$ and $\mathcal{C}^1_n$ is not $\operatorname{poly}(n)$-stabilizable, neither is $\mathcal{C}^2_n$.
    \label{lem1}
\end{lemma}

Lemma \ref{lem1} turns Problem \ref{prob1} into the problem of finding a pair of systems that are not $\operatorname{poly}(n)$-stabilizable. Specifically, if a pair of systems is not $\operatorname{poly}(n)$-stabilizable, then any class containing this pair of systems is also not $\operatorname{poly}(n)$-stabilizable.

{The next two definitions are related to the co-stabilizability and distinguishability of a pair of systems.}

\begin{definition}[Co-stabilizability]
A pair of systems $\mathcal{S}_1=(\A_1, \B_1)$ and $\mathcal{S}_2=(\A_2, \B_2)$ is co-stabilizable if there exists a state-feedback gain $\mathbf{K}$ such that both $\A_1 + \B_1 \mathbf{K}$ and $\A_2 + \B_2 \mathbf{K}$ are stable.
\label{def3}
\end{definition}

\begin{remark}
Co-stabilization problem for two dynamical systems has been studied in robust control \cite{zhou1998essentials}, e.g., by using the gap metric \cite{georgiou1990optimal}.
\end{remark}

We will use $\operatorname{KL}$ divergence to measure the distance between the distributions of state-input trajectories generated when the same exploration policy is applied to two different systems. A small $\operatorname{KL}$ divergence means that it is hard to distinguish two systems.

\begin{definition}[Kullback–Leibler (KL) divergence] The KL divergence between the
continuous distributions $\mathbb{P}$ and $\mathbb{Q}$ is defined as
$$
\operatorname{KL}(\mathbb{P}, \mathbb{Q})= \int_{-\infty}^{+\infty} p(x) \log \frac{p(x)}{q(x)} d x ,
$$
where $p(x)$ and $q(x)$ denote the probability densities of $\mathbb{P}$ and $\mathbb{Q}$ and $p(x)$ is absolutely continuous with respect to $q(x)$.
\label{df1}
\end{definition}

Our main insight behind constructing not $\operatorname{poly}(n)$-stabilizable pairs in the next section is as follows. If we have two different systems and excite all the modes of these systems, as we increase the trajectory length $N$, we expect that the $\operatorname{KL}$ divergence between the trajectories will increase and we will be able to distinguish the systems. On the other hand, if the $\operatorname{KL}$ divergence remains small independent of the exploration policy, then we cannot expect the learning-to-stabilize algorithm to result in significantly different controller gains. Moreover, if these two systems are not co-stabilizable, then learning to stabilize these systems will be hard.

\section{Hard to Learn to Stabilize Systems}
\label{sec3}

Consider the following system of the form (\ref{eq1}) with $(\mathbf{A,B})$ defined parametrically as
\begin{equation}
\begin{aligned}
&\mathbf{A} =
\left[\begin{array}{ccccc}
r & v & 0 & \cdots & 0 \\
0 & 0 & v & \cdots & 0 \\
& & \ddots & \ddots & \\
0 & 0 & 0 & \cdots & v\\
0 & 0 & 0 & \cdots & 0
\end{array}\right] \in \mathbb{R}^{n \times n}, \; \mathbf{B}=
 \left[\begin{array}{c}
b^{(1)} \\
0 \\
\vdots \\
0 \\
v
\end{array}\right] \in \mathbb{R}^{n},
\end{aligned}
\label{eq2}
\end{equation}
where $n\geq 2$, $ r > 1$, $0< v< \frac{r - 1}{2} $, and $b^{(1)} \geq 0$.

\begin{remark}
    When $b^{(1)}=-v^n/r^{n-1}$, the system in \eqref{eq2} is uncontrollable. To avoid this trivially hard-to-stabilize case, we let $b^{(1)}\geq0$.
\end{remark}

The following proposition proves that there exist two systems in the parametric family (\ref{eq2}) differing only in $b^{(1)}$, such that for a feedback gain to be able to stabilize both systems at the same time, the difference in $b^{(1)}$ should be exponentially small in the system dimension.
\begin{proposition} Let $\mathcal{S}_1=(\mathbf{A} ,\mathbf{B}_1)$, and $\mathcal{S}_2=(\mathbf{A} ,\mathbf{B}_2)$, where $\mathbf{A}$ is as in \eqref{eq2}, and $\mathbf{B}_1$   and $\mathbf{B}_2$ equal to $\mathbf{B}$ in \eqref{eq2} with $b^{(1)} = 0$ and $b^{(1)} = m \geq 0$, respectively. Let $\mathbf{K}\in \mathbb{R}^{1\times n}$ be any stabilizing linear state-feedback gain for $\mathcal{S}_1$ such that $\rho(\mathbf{A}+\mathbf{B}_1 \mathbf{K})<1$. Let $   p_{1}^{cl},p_{2}^{cl},\dots,p_{n}^{cl}$ be the eigenvalues of $\mathbf{A}  + \mathbf{B}_1 \mathbf{K}$ with $0\leq |p_{1}^{cl}|,|p_{2}^{cl}|,\dots,|p_{n}^{cl}|<1$.
Then $\rho (\mathbf{A}  + \mathbf{B}_2 \mathbf{K})<1$  only if
\begin{equation}
   0\leq m < v^n \prod_{i=1}^{n} \frac{1+p_i^{cl}}{r -p_{i}^{cl}}.
   \label{prop2-(4)}
\end{equation}
\label{prop2}
\end{proposition}
The proof, which uses Ackermann's formula (Lemma~\ref{lem3}) to analytically compute any stabilizing feedback gain of ($\mathbf{A},\mathbf{B}_1$) and Jury stability test (Lemma~\ref{lem4}) to verify the closed-loop stability of ($\mathbf{A},\mathbf{B}_2$) when using the stabilizing gain of the former, is given in Appendix~\ref{apxA}. 

Next, we upper bound the $\operatorname{KL}$ divergence between the probability distributions of length $N$ input-state trajectories generated by the two LTI systems defined in Proposition \ref{prop2}. Similar upper bounds of the $\operatorname{KL}$ divergence between two LTI systems can also be found in  \cite{jedra2019sample,tsiamis2021linear,tsiamis2022learning}.
\begin{proposition}
Let the systems $\mathcal{S}_1$, and $\mathcal{S}_2$ be the same as those defined in Proposition \ref{prop2}. Let $\pi$ be any learning-to-stabilize algorithm that satisfies Assumption~\ref{asm2}.
Then, the $\operatorname{KL}$ divergence between $ \mathbb{P}^N_{\mathcal{S}_{1},\pi}$ and $ \mathbb{P}^N_{\mathcal{S}_{2},\pi}$ satisfies
$$\operatorname{KL}\left(\mathbb{P}^N_{\mathcal{S}_{1},\pi}, \mathbb{P}^N_{\mathcal{S}_{2},\pi}\right) \leq \frac{N m^2 \sigma_u^2}{2\sigma_w^2 }.$$
\label{prop1}
\end{proposition}
The proof is given in Appendix \ref{secD}.

The next theorem states that there exist some classes of systems with non-degenerate process noise, for which the worst-case sample complexity of learning to stabilize is at least exponential with the system dimension $n$.

\begin{theorem}
\label{thrm_main}
Consider $\mathcal{S}_1$ and $\mathcal{S}_2$ defined in Proposition \ref{prop2}, with $m = 2\left( \frac{2v}{r-1} \right)^n$. Consider any class $\mathcal{C}_n$ of systems  including $\mathcal{S}_1$ and $\mathcal{S}_2$, which satisfies Assumption \ref{asm1}. 
Then, for all learning-to-stabilize algorithms $\pi$ satisfying Assumption \ref{asm2} and
for all confidence levels $0 < \delta< 1/2$, the requirement
\begin{equation}  \inf_{\mathcal{S} \in \mathcal{C}_{n}} \mathbb{P}^N_{\mathcal{S},\pi}
\left( \rho \left(\mathbf{A+B} \pi_N(\ub_{0:N-1}, \x_{1:N}\right) < 1\right) \geq 1-\delta 
\label{eq56}\end{equation}
is satisfied only if
$$
N \geq \frac{ \sigma_w^2}{2 \sigma_u^2  } \left( \frac{r-1 }{2v} \right)^{2n} \log \frac{1}{3 \delta},
$$
where $n \geq 2$, $r > 1$, and $0 < v < \frac{r-1}{2} $.
\label{thm1}
\end{theorem}
The proof of Theorem \ref{thm1} can be found in Appendix \ref{appc}. In the proof we show that if the same algorithm $\pi$ is applied to $\mathcal{S}_1$ and $\mathcal{S}_2$, for the stabilization probability in \eqref{eq56} to be high for both, exponentially many samples are needed. This indicates that for any class containing $\mathcal{S}_1$ and $\mathcal{S}_2$, polynomially many samples will not be sufficient for the satisfaction of requirement \eqref{eq56}, therefore such classes cannot be $\text{poly}(n)$-stabilizable.

Comparing the systems $\mathcal{S}_1$ and $\mathcal{S}_2$ in our proof to corresponding system pairs in \cite{tsiamis2022learning}, our pairs are individually not necessarily ``hard to identify'' but the distance $m$ in the parameter space between the pairs shrinks exponentially fast as we increase $n$. As shown in Proposition~\ref{prop1}, the input-state trajectory distributions our pairs of systems generate look very similar; this is expected since the system parameters get closer with $n$. In general, one may expect if two systems are close to each other in the parameter space, they can be co-stabilized by the same controller $\mathbf{K}$. However, our pairs cannot be co-stabilized (as shown in Proposition~\ref{prop2}) with a single gain $\mathbf{K}$ although the systems are very close in parameter space, which is the main source of hardness.

\begin{remark}
Our proof technique can also be extended to show the hardness of learning to stabilize for classes of systems containing single-input systems with diagonal state matrices and $n$ unstable eigenvalues in a compact range, presented in \cite{li2022fundamental}. In that case, when the input vector is the all-one vector, the controllability matrix of the system is a Vandermonde matrix, which allows us to again use Ackermann's formula to obtain the explicit form of all stabilizing linear state-feedback gains. Results similar to Proposition \ref{prop2} and Theorem \ref{thm1} can be established in this case too.
\end{remark}

\section{Numerical Experiments}
 \label{sec5}
 In this section, we implement two numerical experiments, i.e., certainty equivalent linear quadratic regulator (LQR) and robust control, to show the hardness of stabilization.
\subsection{Certainty Equivalent LQR}
{Since solving LQR problems always gives stabilizing controllers (under mild regularity conditions), the first experiment considers the certainty equivalent LQR control \cite{mania2019certainty}. Specifically, a controller is computed by solving an LQR problem using some estimated system dynamics and then applied to the ground truth system.}
The infinite-horizon LQR problem, simplified as $dLQR(\mathbf{A},\mathbf{B},\mathbf{Q},\mathbf{R})$, is as follows.
\begin{equation}
\begin{aligned}
&\min_{\mathbf{u}_0,\mathbf{u}_1,\cdots}\lim _{T \rightarrow \infty} \mathbb{E} \left[   \frac{1}{T} \sum_{t=0}^{T}\left(\mathbf{x}_{t}^{\top} \mathbf{Q} \mathbf{x}_{t}+ \mathbf{u}_{t}^{\top} \mathbf{R} \mathbf{u}_{t}\right) \right]\\
& \begin{array}{r@{\quad}r@{}l@{\quad}l}
\text{s.t.} & (\ref{eq1}) \\
\end{array} 
\end{aligned} 
\end{equation}
where $\mathbf{Q}, \mathbf{R}$ are positive semi-definite cost matrices.  Its solution is given by $\mathbf{u}_t = \mathbf{K} \mathbf{x}_t$ where the controller $\mathbf{K}$ can be computed by solving the Riccati equation. 
Consider the system $(\mathbf{A},\mathbf{B}_1)$ defined in Proposition \ref{prop2}, and let $\mathbf{Q}=\mathbf{I}_n$, and $\mathbf{R}=1$. 
{Since the analysis of Theorem~\ref{thrm_main} is established on perturbing $b^{(1)}$ in $\mathbf{B}_1$, we consider a simplified setting where only $b^{(1)}$ is unknown and to be estimated using the least squares estimator, which is denoted by $\hat{b}^{(1)}$. Let $\hat{\mathbf{B}}_1$ denote the matrix by replacing $b^{(1)}$ with $\hat{b}^{(1)}$ and $\hat{\mathbf{K}}$ denote the certainty equivalent controller for $(\mathbf{A},\mathbf{B}_1)$ obtained by solving $dLQR(\mathbf{A},\hat{\mathbf{B}}_1,\mathbf{Q},\mathbf{R})$. We let the input $u_t \overset{i.i.d.}{\sim} \mathcal{N}(0, \sigma_u^2)$. Since there is only a single unknown parameter and its regressor $u_t$ is independent, this system is trivially easy to identify.

For each dimension $n$, we run $M = 200$ independent experiments. Let $\hat{\mathbf{K}}_{i,N'}$ denote the controller obtained using the first $N'$ data points, i.e., $\{ \mathbf{u}_{0:N'-1}, \mathbf{x}_{1:N'} \}$, in the $i^{th}$ experiment. We record the smallest trajectory length $N$ under which at least 90\% of the experiments produce stabilizing controllers, i.e.
\begin{equation}
    N : = \min \Big\{ N' \in \mathbb{N} : \frac{1}{M} \sum_{i \in [M]} \mathbb{I}_{ \{\rho(\mathbf{A}+\mathbf{B}_1 \hat{\mathbf{K}}_{i,N'})<1 \} } \geq 0.9 \Big\},
\end{equation}
}
where $\mathbb{I}$ denotes the indicator function.

The results are given in Fig.~\ref{fig1}. According to Fig.~\ref{fig1}, we have that as the system dimension increases, the required number of samples for a given frequency of stability increases exponentially with the system dimension. 

\begin{figure}
  \centering
  \includegraphics[scale=0.38]{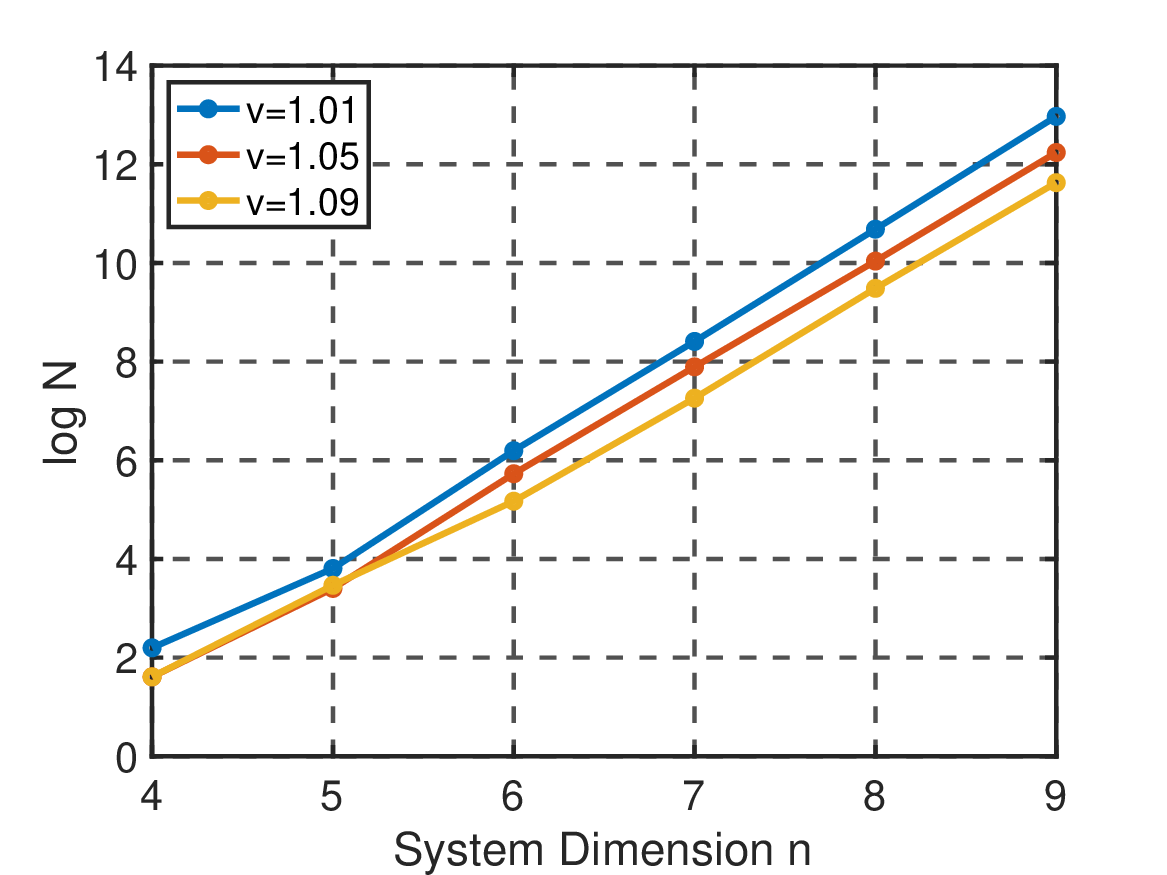}
  \caption{{Required trajectory length $N$ for $90\%$ stabilization rate vs. system dimension $n$: $\sigma_u^2 = 32$, $\sigma_w^2 = 0.005$, $v \in \{1.01, 1.05, 1.09\}$}}
  \label{fig1}
\vspace{-0.3cm}
\end{figure}

\subsection{LMI-based Sufficient Condition for Co-stabilizability}

 In this section, we numerically demonstrate the hardness of co-stabilizability of $\mathcal{S}_1=(\mathbf{A},\mathbf{B}_1)$ and $\mathcal{S}_2=(\mathbf{A},\mathbf{B}_2(m))$ using ideas from robust control, where we leave $m$ as a parameter. We use the following feasibility problem, which can be converted to an LMI, to check sufficient conditions of co-stabilizability.
\begin{equation}
\begin{array}{cl}
\text { find }  & \mathbf{K, P} \\
\text {s.t.} & \left( \mathbf{A} + \mathbf{B}_1 \mathbf{K} \right)^{\top}  \mathbf{P} \left( \mathbf{A} + \mathbf{B}_1 \mathbf{K} \right) \prec \mathbf{P}\\
&\left( \mathbf{A} + \mathbf{B}_2(m) \mathbf{K} \right)^{\top}  \mathbf{P} \left( \mathbf{A} + \mathbf{B}_2(m)\mathbf{K} \right) \prec \mathbf{P}\\
& \mathbf{P} \succ 0 \\
\end{array}.
\label{eq13}
\end{equation}
We use the bisection method to find the largest $m$ such that the problem (\ref{eq13}) is feasible. The results are shown in Fig.~\ref{fig2}. According to this figure, we see that as the system dimension increases, the largest $m$ such that the LMI optimization problem in (\ref{eq13}) is feasible decreases exponentially with increasing system dimension, which is consistent with Eq.~\eqref{prop2-(4)} in Proposition \ref{prop2}. 

\begin{figure}
  \centering
  \includegraphics[scale=0.38]{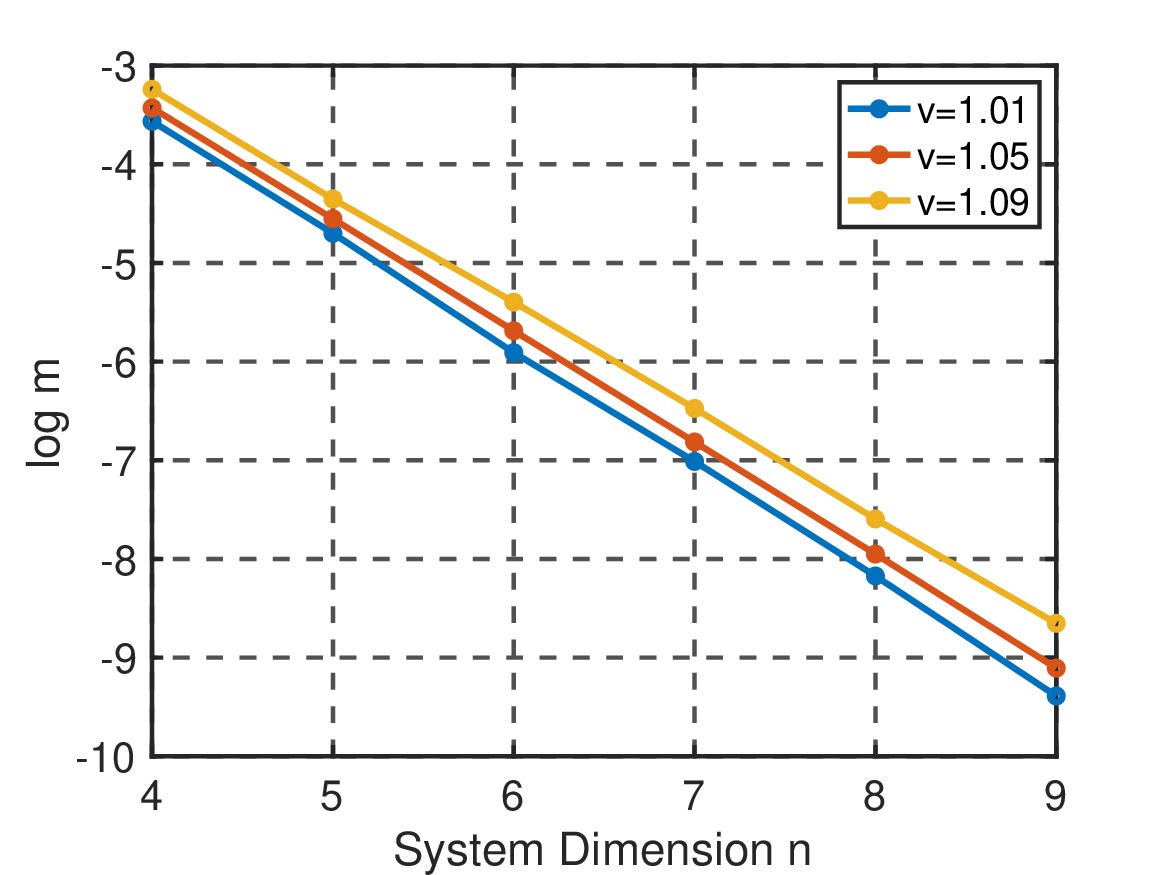}
  \caption{$v=1.01,1.05,1.09$, and $r = 3.2$. The $x$-axis is the system dimension $n$ and the $y$-axis is the logarithm of the largest $m$ such that the problem in (\ref{eq13}) is feasible.}
  \label{fig2}
\vspace{-0.3cm}
\end{figure}

\section{Conclusion and Future Work}
In this work, we identified an extended class of LTI systems that are hard to learn to stabilize with static state feedback. The main idea in constructing such examples is to find pairs of systems whose parameters become exponentially close to each other as the dimension increases, yet they are not co-stabilizable. One interesting observation is that the entries of stabilizing gains for these pairs are also growing exponentially (see, Eq.~\eqref{eq36}). In the future, we want to investigate the ramifications of this observation in gradient-based learning algorithms used for control as in \cite{ziemann2022policy}.  

\noindent\emph{Acknowledgments: }The authors would like to thank Prof. Peter Seiler of University of Michigan for some early discussions that motivated this work.

  \appendix
  
\subsection{Proof of Proposition \ref{prop2}}
\label{apxA}
We first introduce a few lemmas used in the proof of Proposition \ref{prop2}. The first lemma parameterizes all stabilizing state-feedback gains for single-input controllable LTI systems. Recall that the controllability matrix of a system $(\A, \B)$ is defined by
\begin{equation}
\mathbf{Ctr}_{(\mathbf{A},\mathbf{B} )}=\left[\begin{array}{llll}
\mathbf{B} & \mathbf{A} \mathbf{B} & \cdots & \mathbf{A}^{n-1} \mathbf{B}
\end{array}\right].
\label{eq26}
\end{equation}

\begin{lemma}[Ackermann's formula \cite{ackermann1972entwurf}
]
 Consider the following order $n$ single-input controllable system $(\mathbf{A},\mathbf{B})$ with state feedback $\mathbf{K} \in \mathbb{R}^{1\times n}$:
$$\left\{ \begin{aligned}
      \mathbf{x}_{t+1} & =\mathbf{A} \mathbf{x}_t + \mathbf{B} \mathbf{u}_t \\  
      \mathbf{u}_t & = \mathbf{K} \mathbf{x}_t
\end{aligned}. \right.
$$
Given $n$ desired eigenvalues of $\mathbf{A}+\mathbf{B}\mathbf{K}$, the unique state feedback that achieves these closed-loop eigenvalues is:
\begin{equation}
\mathbf{K} = -\mathbf{e}_n^{\mathrm{\top}} \mathbf{Ctr}^{-1}_{(\mathbf{A},\mathbf{B})} \Delta^{\mathrm{cl}}(\mathbf{A}),
\label{eq28}
\end{equation}
where $\mathbf{e}_n$ is the last column of the $n \times n$ identity matrix, and $\Delta^{\mathrm{cl}}(\mathbf{A})$ is the characteristic polynomial of $\mathbf{A+BK}$ evaluated at $\mathbf{A}$.
     \label{lem3}
 \end{lemma} 
The next lemma derives the expression of the first element of any stabilizing state-feedback gains for $(\mathbf{A},\mathbf{B}_1)$, parameterized by the stable closed-loop poles. 
\begin{lemma}
 For the system $(\mathbf{A},  \mathbf{B}_1 )$ defined in Proposition \ref{prop2} and any stabilizing state feedback $\mathbf{K}\in \R^{1\times n}$, let $\{   p_{k}^{cl}\}_{k=1}^n$ be the eigenvalues of $\mathbf{A}  + \mathbf{B}_1 \mathbf{K}$, with $\{   p_{k}^{cl}\}_{k=1}^n$ all inside the unit circle.
 Then, the first element $k_1$ of  the state feedback $\mathbf{K}$ satisfies
 \begin{equation}
k_1 =  -\frac{\left(r -p_{1}^{cl} \right)\left(r -p_{2}^{cl} \right)\cdots \left(r -p_{n}^{cl} \right)}{v^n}.
     \label{eq39}
 \end{equation}
 
     \label{lem7}
 \end{lemma} 
\begin{proof}
By Lemma \ref{lem3}, since $(\mathbf{A}, \mathbf{B}_1)$ is single-input and controllable, the state feedback $\mathbf{K}$ satisfies 
\begin{equation}
\mathbf{K}= -\mathbf{e}_n^{\mathrm{\top}} \mathbf{Ctr}_{(\mathbf{A}, \mathbf{B}_1)}^{-1} \Delta^{\mathrm{cl}}(\mathbf{A}),
\label{eq48}
\end{equation}
where the characteristic polynomial $\Delta^{\mathrm{cl}}(\A)$ of the closed-loop system $\mathbf{A}+\mathbf{B}_1 \mathbf{K}$ evaluated at $\A$ is
\begin{equation}
\begin{aligned}
\Delta^{\mathrm{cl}}(\mathbf{A}) =\prod_{i=1}^{n} (\mathbf{A} -p_{i}^{cl} \mathbf{I}_n ).
\end{aligned}
\label{eq34}
\end{equation}
By the definition of $\mathbf{A}$ in \eqref{eq2}, for all $i=1$, $2$, ..., $n$,
\begin{equation}
\mathbf{A}-p_{i}^{cl} \mathbf{I}_n =  \left[\begin{array}{ccccc}
r -p_{i}^{cl} & v & 0 & \cdots & 0 \\
0 & -p_{i}^{cl} & v & \cdots & 0 \\
& & \ddots & \ddots & \\
0 & 0 & 0 & \cdots & v\\
0 & 0 & 0 & \cdots & -p_{i}^{cl}
\end{array}\right].
    \label{eq51}
\end{equation}
Based on \eqref{eq51}, the element of $\Delta^{\mathrm{cl}}(\mathbf{A})$ at the first row and the first column is
\begin{equation} \left[ \Delta^{\mathrm{cl}}(\mathbf{A}) \right]^{(1,1)}  = \left(r -p_{1}^{cl} \right)\left(r -p_{2}^{cl} \right)\cdots \left(r -p_{n}^{cl} \right).
\label{eq35}
\end{equation}
Furthermore, due to the special structures of  $(\mathbf{A},  \mathbf{B}_1 )$, it can be shown that the last row of the inverse of the controllability matrix $\mathbf{Ctr}_{(\mathbf{A},\mathbf{B}_1)}$ is
\begin{equation}\mathbf{e}_n^{\mathrm{\top}} \mathbf{Ctr}^{-1}_{(\mathbf{A},\mathbf{B}_1)} =\left[\begin{array}{cccc}   
    v^{-n} & 0 & \cdots & 0 \\  
  \end{array}\right].
  \label{eq33}
  \end{equation}
Thus, according to (\ref{eq48}), \eqref{eq35}, and \eqref{eq33},  the first element of the state feedback $\mathbf{K}$ is
\begin{equation}
\begin{aligned}
k_1 &= -v^{-n}  \left[ \Delta^{\mathrm{cl}}(\mathbf{A}) \right]^{(1,1)} \\
&=  -\frac{\left(r -p_{1}^{cl} \right)\left(r -p_{2}^{cl} \right)\cdots \left(r -p_{n}^{cl} \right)}{v^n}.
\end{aligned}
\label{eq36}
\end{equation}
\end{proof}

The next lemma provides a necessary condition for the stability of discrete-time LTI systems.
 \begin{lemma}[Jury stability test, Theorem 4.6 in \cite{fadali2012digital}]
 For the polynomial
$$
\Delta(z)=a_n z^n+a_{n-1} z^{n-1}+\ldots+a_1 z+a_0=0, 
$$
with $a_n>0$, the roots of the polynomial are inside the unit circle only if
\begin{equation}
    \begin{cases}
       \Delta(1)>0 \\
   (-1)^n \Delta(-1)>0 
\end{cases}.
\label{eq27}
\end{equation}
     \label{lem4}
 \end{lemma}

 Now, we are ready to present the proof of Proposition \ref{prop2}.
 
\begin{proof}
 Consider the two systems $(\mathbf{A},\mathbf{B}_1)$ and $(\mathbf{A},\mathbf{B}_2)$ in Proposition \ref{prop2}. 
 
Let $\mathbf{K}$ be any stabilizing state-feedback gain of $(\A, \B_1)$.
By Lemma \ref{lem7},
 the first element of $\mathbf{K}$  satisfies
 \begin{equation}
k_1 =  -\frac{\left(r -p_{1}^{cl} \right)\left(r -p_{2}^{cl} \right)\cdots \left(r -p_{n}^{cl} \right)}{v^n},
     \label{eq53}
 \end{equation}
 where $p_{1}^{cl}$, $p_{2}^{cl}$, ..., $p_{n}^{cl}$ are the eigenvalues of $\mathbf{A}  + \mathbf{B}_1 \mathbf{K}$ with $|p_{1}^{cl}|,|p_{2}^{cl}|,\dots,|p_{n}^{cl}|<1$.
Next, it can be shown that the characteristic polynomial $\Delta_n^{\mathrm{cl}}(z)$ of $\mathbf{A} + \mathbf{B}_1 \mathbf{K}$ is 
\begin{equation}
\begin{aligned}
\Delta_n^{\mathrm{cl}}(z) :=& \det(z \mathbf{I-\mathbf{A} - \mathbf{B}_1   \mathbf{K}})\\
 = & z^n- \left(r + v k_n \right)z^{n-1} \\
 &+  \sum_{j=0}^{n-2} v^{n-j-1} \left(r k_{j+2} - v k_{j+1} \right) z^{j} 
\end{aligned}.
\label{eq54}
\end{equation}
Similarly, one can show that the characteristic polynomial $\hat{\Delta}_n^{\mathrm{cl}}(z)$ of $\mathbf{ A} + \mathbf{B}_2  \mathbf{K}$ satisfies
\begin{equation}
\begin{aligned}
\hat{\Delta}_n^{\mathrm{cl}}(z) &:= \det(z \mathbf{I-\mathbf{ A} - \mathbf{B}_2  \mathbf{K}   })\\
=& z^n-(r + v k_n + k_1 m)z^{n-1} + \\ & \sum_{j=0}^{n-2} v^{n-j-1} (r k_{j+2} - v k_{j+1}) z^{j} \\
&= \Delta_n^{\mathrm{cl}}(z) - mk_1 z^{n-1}.
\end{aligned}
\label{eq70}
\end{equation}
Thus, by \eqref{eq70}, we have
\begin{equation}
\begin{cases}
    \hat{\Delta}_n^{\mathrm{cl}}(1) =\Delta_n^{\mathrm{cl}}(1) - mk_1, \\
    \hat{\Delta}_n^{\mathrm{cl}}(-1) =\Delta_n^{\mathrm{cl}}(-1) - (-1)^{n-1}mk_1.
    \end{cases}
    \label{eq55}
\end{equation}
By Lemma \ref{lem4}, the matrix $\mathbf{ A} + \mathbf{B}_2  \mathbf{K}$ is stable only if
\begin{equation}\label{eq23}
\begin{cases}
    \hat{\Delta}_n^{\mathrm{cl}}(1) > 0\\
 (-1)^{n}   \hat{\Delta}_n^{\mathrm{cl}}(-1) > 0
    \end{cases}.
\end{equation}
Also, note that
\begin{equation}
\begin{cases}
\Delta_n^{\mathrm{cl}}(1) = \prod_{i=1}^{n}(1-p_i^{cl})\\
\Delta_n^{\mathrm{cl}}(-1) = (-1)^n \prod_{i=1}^{n}(1+p_i^{cl})
\end{cases}.
\label{eq52}
\end{equation}
Combining (\ref{eq55}), (\ref{eq23}), and (\ref{eq52}), we have that
$\mathbf{ A} + \mathbf{B}_2  \mathbf{K}$ is stable only if
\begin{equation}
    0 \leq m < v^n \prod_{i=1}^{n} \frac{1+p_i^{cl}}{r -p_{i}^{cl}},
\end{equation}
with $|p_{1}^{cl}|,|p_{2}^{cl}|,\dots,|p_{n}^{cl}|<1$.
\end{proof}

\subsection{Proof of Proposition \ref{prop1}}
\label{secD}

For simplicity of notation, we denote  $\mathbb{P}_{\mathcal{S}_i,\pi}^t$, $f_{\mathcal{S}_i,\pi}^t$, and $\mathbb{E}_{\mathcal{S}_1,\pi}^t$ by $\mathbb{P}^t_i$,  $f^t_i$, and  $\mathbb{E}^t_1$ respectively, for $i=1,2$, and $0\leq t\leq N$. With this notation,
Proposition \ref{prop1} can be proven as follows.
\begin{proof} Starting with the definition of KL divergence (i.e., Definition \ref{df1}), we have
{\small
\begin{align}\label{eq69}
\begin{split}
\operatorname{KL}\left(\mathbb{P}^N_1, \mathbb{P}^N_2\right)&  {=} \mathbb{E}^N_{1} \left[ \log \frac{ f^N_1 \left( \mathbf{u}_{0:N-1},\mathbf{x}_{0:N} \right) }{ f^N_2 \left( \mathbf{u}_{0:N-1},\mathbf{x}_{0:N} \right) } \right] 
\\
&=\mathbb{E}^N_{1} \left[ \log \frac{ \prod_{t=0}^{N} f^t_1 \left( \mathbf{x}_t \mid \x_{0: t-1}, \mathbf{u}_{0: t-1} \right)  }{ \prod_{t=0}^{N} f^t_2 \left( \mathbf{x}_t \mid \x_{0: t-1}, \mathbf{u}_{0: t-1} \right)   } \right]\\
&\quad +\mathbb{E}^N_{1} \left[ \log \frac{   \prod_{t=0}^{N-1} f^t_1 \left( \mathbf{u}_t \mid \x_{0: t}, \mathbf{u}_{0: t-1} \right) }{   \prod_{t=0}^{N-1} f^t_2 \left( \mathbf{u}_t \mid \x_{0: t}, \mathbf{u}_{0: t-1} \right) } \right]
\\
&=  \sum_{t=0}^N \mathbb{E}^t_{1} \left[ \log \frac{   f^t_1 \left( \mathbf{x}_t \mid \x_{  t-1}, \mathbf{u}_{  t-1} \right)  }{   f^t_2 \left( \mathbf{x}_t \mid \x_{  t-1}, \mathbf{u}_{  t-1} \right)   } \right],
\end{split}
\end{align}
}\noindent where the second equality is from the properties of the conditional probability density functions and the third equality is because the exploration policies of these two systems are the same and the discrete-time LTI system has the Markovian structure.

Based on the special structure of $ (\mathbf{A,B})$ of $\mathcal{S}_1$ and $\mathcal{S}_2$, we have the following relationships between every element of state vectors of these systems:
\begin{equation} \mathcal{S}_1:
\begin{cases}
x_t^{(1)} = r x_{t-1}^{(1)} + v x_{t-1}^{(2)}  + w^{(1)}_{t-1}\\
x_t^{(j)} = v x_{t-1}^{(j+1)} + w^{(j)}_{t-1}, \;\text{for} \; j=2,\dots,n-1\\
x_t^{(n)} = v \ub_{t-1} + w^{(n)}_{t-1}
\end{cases},
\label{eq144}
\end{equation}
\begin{equation}
\mathcal{S}_2 :
\begin{cases}
x_t^{(1)} = r x_{t-1}^{(1)} + v x_{t-1}^{(2)} +    m  \ub_{t-1}  + w^{(1)}_{t-1} \\
x_t^{(j)} = v x_{t-1}^{(j+1)} + w^{(j)}_{t-1}, \;\text{for} \; j=2,\dots,n-1\\
x_t^{(n)} = v \ub_{t-1} + w^{(n)}_{t-1}
\end{cases}.
\label{eq65}
\end{equation}
Due to (\ref{eq144}), (\ref{eq65}) and the fact that $w^{(j)}_{t}$ for $j=1,\dots,n$ are mutually independent, 
we   have for $i=1$ and $2$, 
\begin{equation}
\begin{aligned}
&f^t_{i}\left(\mathbf{x}_{t} \mid \mathbf{x}_{t-1},\ub_{t-1}\right)\\
=& f^t_{i}\left(x_{t}^{(1)} \mid x_{t-1}^{(1)},x_{t-1}^{(2)}, \ub_{t-1}\right)   \cdot \\
&\prod^{n-1}_{j=2} f^t_{i}\left(x_{t}^{(j)}  \mid x_{t-1}^{(j+1)} \right) f^t_{i}\left(x_t^{(n)}\mid \ub_{t-1}\right).
\end{aligned}
\label{eq66}
\end{equation}
According to (\ref{eq144}) and (\ref{eq65}), we also   have
\begin{equation}
\begin{cases}
\begin{aligned}
&f^t_{1}\left(x_{t}^{(j)}  \mid x_{t-1}^{(j+1)} \right) = f^t_{2}\left(x_{t}^{(j)}  \mid x_{t-1}^{(j+1)} \right) \\
&\quad \;\text{for} \; j=2,\dots,n-1 
\end{aligned}\\
f^t_{1}\left(x_t^{(n)}\mid \ub_{t-1}\right) = f^t_{2}\left(x_t^{(n)}\mid \ub_{t-1}\right)
\end{cases},
\label{eq67}
\end{equation} 
and
\begin{equation}
\begin{cases}
\begin{aligned}
&f^t_{1}\left(x_{t}^{(1)} \mid x_{t-1}^{(1)},x_{t-1}^{(2)},\ub_{t-1}\right) \sim \\
& \quad \mathcal{N} \left(r x_{t-1}^{(1)} + v x_{t-1}^{(2)} ,  \sigma^2_{w}\right) 
\end{aligned}\\
\begin{aligned}
& f^t_{2}\left(x_{t}^{(1)} \mid x_{t-1}^{(1)},x_{t-1}^{(2)},\ub_{t-1}\right) \sim \\
& \quad \mathcal{N} \left(r x_{t-1}^{(1)} + v x_{t-1}^{(2)} + m \ub_{t-1},   \sigma^2_{w} \right)
\end{aligned}
\end{cases},
    \label{eq75}
\end{equation}\noindent{where $\mathcal{N} $ denotes the Gaussian distribution.}

Then, (\ref{eq69})    is equal to
{\small
\begin{align}    \label{eq68}
\begin{split}
\operatorname{KL}\left(\mathbb{P}^N_{1}, \mathbb{P}^N_{2}\right)&= 
 \sum^{N}_{t=1} \mathbb{E}^t_{1} \left[ \log \frac{ f^t_{1}\left(x_{t}^{(1)} \mid x_{t-1}^{(1)},x_{t-1}^{(2)}, \ub_{t-1}\right) }{ f^t_{2}\left(x_{t}^{(1)} \mid x_{t-1}^{(1)},x_{t-1}^{(2)}, \ub_{t-1}\right)}   \right]  \\
 &= \sum^{N}_{t=1} \mathbb{E}^{t}_{1} \left[  \frac{ \left( w^{(1)}_{t-1} + m \ub_{t-1} \right)^2 - \left( w^{(1)}_{t-1}\right)^2 }{2  \sigma^2_{w} } \right]\\
&= \sum^{N}_{t=1} \mathbb{E}^{t}_{1} \left[  \frac{\left( m \ub_{t-1} \right)^2}{2  \sigma^2_{w} } \right]\\
&\leq \frac{Nm^2 \sigma_u^2}{2\sigma_w^2 } ,
\end{split}
\end{align}
}\noindent{where the first equality is due to (\ref{eq66}) and (\ref{eq67}), the second equality is due to \eqref{eq144}, (\ref{eq75}), and {the definition of the Gaussian distribution, the third equality is by the noise process being zero mean and $w^{(1)}_{t-1}$ and $\ub_{t-1}$} being independent, and the last inequality is due to Assumption \ref{asm2}.}
\end{proof}

\subsection{Proof of Theorem \ref{thm1}}
\label{appc}
Before presenting the proof of Theorem \ref{thm1}, we first introduce Birgé's inequality, a classical inequality from information theory.
\begin{lemma}[Birgé's Inequality, Theorem 4.21 in \cite{boucheron2013concentration}]Let $\Omega$ be a set and $\mathcal{E}$ be a $\sigma$-algebra on the set $\Omega$. Let $\mathbb{P}_1, \mathbb{P}_2$ be probability measures on the probability space $(\Omega, \mathcal{E})$ and let $E_1, E_2 \in \mathcal{E}$ be disjoint events. If $1-\delta \triangleq \min _{i=1,2} \mathbb{P}_i\left(E_i\right) \geq 1 / 2$ then
$$
\operatorname{KL}\left(\mathbb{P}_1, \mathbb{P}_2\right) \geq (1-\delta) \log \frac{1-\delta}{\delta}+\delta \log \frac{\delta}{1-\delta}.
$$
    \label{lem9}
\end{lemma}
 Next, we present the proof of Theorem \ref{thm1}.
 
\begin{proof}
Given $\mathcal{S}_1$ and $\mathcal{S}_2$, let us define two events:
\begin{align}
\begin{split}
    \nonumber E_1 &= \left\{ \ub_{0:N-1}, \x_{1:N} \mid  \rho \left(\mathbf{A}+\mathbf{B}_1 \pi_N(\ub_{0:N-1}, \x_{1:N})\right) < 1  \right\}, \\
 \nonumber  E_2 &= \left\{    \ub_{0:N-1}, \x_{1:N} \mid   \rho \left(\mathbf{A}+\mathbf{B}_2 \pi_N(\ub_{0:N-1}, \x_{1:N})\right) < 1 \right\}.
\end{split}
\label{eq61}
\end{align}

Since $m = 2\left( \frac{2v}{r-1} \right)^n > v^n \prod_{i=1}^{n} \frac{1+p_i^{cl}}{r -p_{i}^{cl}}$ for any stable closed-loop poles, by Proposition~\ref{prop2}, $\mathcal{S}_1$ and $\mathcal{S}_2$ cannot be co-stabilized.
Hence, $E_1$ and $E_2$ are disjoint events.

Suppose \eqref{eq56} is true, which implies
\begin{equation}
    \begin{cases}
\mathbb{P}^N_{\mathcal{S}_1,\pi}
\left( E_1\right) \geq 1-\delta,\\
\mathbb{P}^N_{\mathcal{S}_2,\pi}
\left( E_2\right) \geq 1-\delta.
    \end{cases}
    \label{eq60}
\end{equation}

Therefore, we can apply Lemma \ref{lem9} to obtain
\begin{equation}
\begin{aligned}
    \operatorname{KL}\left(  \mathbb{P}^N_{\mathcal{S}_1,\pi},   \mathbb{P}^N_{\mathcal{S}_2,\pi} \right) &\geq (1-\delta) \log \frac{1-\delta}{\delta}+\delta \log \frac{\delta}{1-\delta}\\
    & \geq \log \left( \frac{1}{3 \delta} \right).
    \end{aligned}
    \label{eq63}
\end{equation}     
According to Proposition \ref{prop1}, the KL divergence between $ \mathbb{P}^N_{\mathcal{S}_{1},\pi}$ and $ \mathbb{P}^N_{\mathcal{S}_{2},\pi}$  satisfies
\begin{equation}\operatorname{KL}\left(\mathbb{P}^N_{\mathcal{S}_{1},\pi}, \mathbb{P}^N_{\mathcal{S}_{2},\pi}\right) \leq \frac{N m^2 \sigma_u^2}{2\sigma_w^2 } \leq \frac{2N\ \sigma_u^2}{\sigma_w^2 } \left( \frac{2v}{r-1} \right)^{2n}.
\label{eq57}
\end{equation}
Combining (\ref{eq57}) and (\ref{eq63}), we have that (\ref{eq56}) holds only if
\begin{equation}
N\geq  \frac{ \sigma_w^2}{2 \sigma_u^2  } \left( \frac{r-1 }{2v} \right)^{2n} \log \left( \frac{1}{3 \delta} \right).
    \label{eq59}
\end{equation}
\end{proof}

\bibliographystyle{IEEEtran}
\bibliography{references}

% Generated by IEEEtran.bst, version: 1.14 (2015/08/26)
\begin{thebibliography}{10}
\providecommand{\url}[1]{#1}
\csname url@samestyle\endcsname
\providecommand{\newblock}{\relax}
\providecommand{\bibinfo}[2]{#2}
\providecommand{\BIBentrySTDinterwordspacing}{\spaceskip=0pt\relax}
\providecommand{\BIBentryALTinterwordstretchfactor}{4}
\providecommand{\BIBentryALTinterwordspacing}{\spaceskip=\fontdimen2\font plus
\BIBentryALTinterwordstretchfactor\fontdimen3\font minus
  \fontdimen4\font\relax}
\providecommand{\BIBforeignlanguage}[2]{{%
\expandafter\ifx\csname l@#1\endcsname\relax
\typeout{** WARNING: IEEEtran.bst: No hyphenation pattern has been}%
\typeout{** loaded for the language `#1'. Using the pattern for}%
\typeout{** the default language instead.}%
\else
\language=\csname l@#1\endcsname
\fi
#2}}
\providecommand{\BIBdecl}{\relax}
\BIBdecl

\bibitem{tsiamis2022learning}
A.~Tsiamis, I.~M. Ziemann, M.~Morari, N.~Matni, and G.~J. Pappas, ``Learning to
  control linear systems can be hard,'' in \emph{Conference on Learning
  Theory}.\hskip 1em plus 0.5em minus 0.4em\relax PMLR, 2022, pp. 3820--3857.

\bibitem{liu2022stability}
T.~Liu, Y.~Song, L.~Zhu, and D.~J. Hill, ``Stability and control of power
  grids,'' \emph{Annual Review of Control, Robotics, and Autonomous Systems},
  vol.~5, pp. 689--716, 2022.

\bibitem{brunke2022safe}
L.~Brunke, M.~Greeff, A.~W. Hall, Z.~Yuan, S.~Zhou, J.~Panerati, and A.~P.
  Schoellig, ``Safe learning in robotics: From learning-based control to safe
  reinforcement learning,'' \emph{Annual Review of Control, Robotics, and
  Autonomous Systems}, vol.~5, pp. 411--444, 2022.

\bibitem{schwarting2018planning}
W.~Schwarting, J.~Alonso-Mora, and D.~Rus, ``Planning and decision-making for
  autonomous vehicles,'' \emph{Annual Review of Control, Robotics, and
  Autonomous Systems}, vol.~1, pp. 187--210, 2018.

\bibitem{cohen2019learning}
A.~Cohen, T.~Koren, and Y.~Mansour, ``Learning linear-quadratic regulators
  efficiently with only $\sqrt{T} $ regret,'' in \emph{International Conference
  on Machine Learning}.\hskip 1em plus 0.5em minus 0.4em\relax PMLR, 2019, pp.
  1300--1309.

\bibitem{simchowitz2020improper}
M.~Simchowitz, K.~Singh, and E.~Hazan, ``Improper learning for non-stochastic
  control,'' in \emph{Conference on Learning Theory}.\hskip 1em plus 0.5em
  minus 0.4em\relax PMLR, 2020, pp. 3320--3436.

\bibitem{zheng2021sample}
Y.~Zheng, L.~Furieri, M.~Kamgarpour, and N.~Li, ``Sample complexity of linear
  quadratic gaussian ({LQG}) control for output feedback systems,'' in
  \emph{Learning for dynamics and control}.\hskip 1em plus 0.5em minus
  0.4em\relax PMLR, 2021, pp. 559--570.

\bibitem{ouyang2019posterior}
Y.~Ouyang, M.~Gagrani, and R.~Jain, ``Posterior sampling-based reinforcement
  learning for control of unknown linear systems,'' \emph{IEEE Transactions on
  Automatic Control}, vol.~65, no.~8, pp. 3600--3607, 2019.

\bibitem{faradonbeh2020optimism}
M.~K.~S. Faradonbeh, A.~Tewari, and G.~Michailidis, ``Optimism-based adaptive
  regulation of linear-quadratic systems,'' \emph{IEEE Transactions on
  Automatic Control}, vol.~66, no.~4, pp. 1802--1808, 2020.

\bibitem{dean2018regret}
S.~Dean, H.~Mania, N.~Matni, B.~Recht, and S.~Tu, ``Regret bounds for robust
  adaptive control of the linear quadratic regulator,'' \emph{Advances in
  Neural Information Processing Systems}, vol.~31, 2018.

\bibitem{sattar2021identification}
Y.~Sattar, Z.~Du, D.~A. Tarzanagh, L.~Balzano, N.~Ozay, and S.~Oymak,
  ``Identification and adaptive control of markov jump systems: Sample
  complexity and regret bounds,'' \emph{arXiv preprint arXiv:2111.07018}, 2021.

\bibitem{fazel2018global}
M.~Fazel, R.~Ge, S.~Kakade, and M.~Mesbahi, ``Global convergence of policy
  gradient methods for the linear quadratic regulator,'' in \emph{International
  Conference on Machine Learning}.\hskip 1em plus 0.5em minus 0.4em\relax PMLR,
  2018, pp. 1467--1476.

\bibitem{cassel2020logarithmic}
A.~Cassel, A.~Cohen, and T.~Koren, ``Logarithmic regret for learning linear
  quadratic regulators efficiently,'' in \emph{International Conference on
  Machine Learning}.\hskip 1em plus 0.5em minus 0.4em\relax PMLR, 2020, pp.
  1328--1337.

\bibitem{zhang2020policy}
K.~Zhang, B.~Hu, and T.~Basar, ``Policy optimization for $\mathcal{H}_2$ linear
  control with $\mathcal{H}_{\infty}$ robustness guarantee: Implicit
  regularization and global convergence,'' in \emph{Learning for Dynamics and
  Control}.\hskip 1em plus 0.5em minus 0.4em\relax PMLR, 2020, pp. 179--190.

\bibitem{tang2021analysis}
Y.~Tang, Y.~Zheng, and N.~Li, ``Analysis of the optimization landscape of
  linear quadratic gaussian ({LQG}) control,'' in \emph{Learning for Dynamics
  and Control}.\hskip 1em plus 0.5em minus 0.4em\relax PMLR, 2021, pp.
  599--610.

\bibitem{li2021distributed}
Y.~Li, Y.~Tang, R.~Zhang, and N.~Li, ``Distributed reinforcement learning for
  decentralized linear quadratic control: A derivative-free policy optimization
  approach,'' \emph{IEEE Transactions on Automatic Control}, 2021.

\bibitem{mohammadi2021convergence}
H.~Mohammadi, A.~Zare, M.~Soltanolkotabi, and M.~R. Jovanovi{\'c},
  ``Convergence and sample complexity of gradient methods for the model-free
  linear--quadratic regulator problem,'' \emph{IEEE Transactions on Automatic
  Control}, vol.~67, no.~5, pp. 2435--2450, 2021.

\bibitem{yang2019provably}
Z.~Yang, Y.~Chen, M.~Hong, and Z.~Wang, ``Provably global convergence of
  actor-critic: A case for linear quadratic regulator with ergodic cost,''
  \emph{Advances in neural information processing systems}, vol.~32, 2019.

\bibitem{abbasi2011regret}
Y.~Abbasi-Yadkori and C.~Szepesv{\'a}ri, ``Regret bounds for the adaptive
  control of linear quadratic systems,'' in \emph{Proceedings of the 24th
  Annual Conference on Learning Theory}.\hskip 1em plus 0.5em minus 0.4em\relax
  JMLR Workshop and Conference Proceedings, 2011, pp. 1--26.

\bibitem{faradonbeh2018finite}
M.~K.~S. Faradonbeh, A.~Tewari, and G.~Michailidis, ``Finite-time adaptive
  stabilization of linear systems,'' \emph{IEEE Transactions on Automatic
  Control}, vol.~64, no.~8, pp. 3498--3505, 2018.

\bibitem{lale2022reinforcement}
S.~Lale, K.~Azizzadenesheli, B.~Hassibi, and A.~Anandkumar, ``Reinforcement
  learning with fast stabilization in linear dynamical systems,'' in
  \emph{International Conference on Artificial Intelligence and
  Statistics}.\hskip 1em plus 0.5em minus 0.4em\relax PMLR, 2022, pp.
  5354--5390.

\bibitem{chen2021black}
X.~Chen and E.~Hazan, ``Black-box control for linear dynamical systems,'' in
  \emph{Conference on Learning Theory}.\hskip 1em plus 0.5em minus 0.4em\relax
  PMLR, 2021, pp. 1114--1143.

\bibitem{hu2022sample}
Y.~Hu, A.~Wierman, and G.~Qu, ``On the sample complexity of stabilizing {LTI}
  systems on a single trajectory,'' \emph{Advances in Neural Information
  Processing Systems}, vol.~35, pp. 16\,989--17\,002, 2022.

\bibitem{dai2020data}
T.~Dai, M.~Sznaier, and B.~R. Solvas, ``Data-driven quadratic stabilization of
  continuous lti systems,'' \emph{IFAC-PapersOnLine}, vol.~53, no.~2, pp.
  3965--3970, 2020.

\bibitem{perdomo2021stabilizing}
J.~Perdomo, J.~Umenberger, and M.~Simchowitz, ``Stabilizing dynamical systems
  via policy gradient methods,'' \emph{Advances in Neural Information
  Processing Systems}, vol.~34, pp. 29\,274--29\,286, 2021.

\bibitem{de2019formulas}
C.~De~Persis and P.~Tesi, ``Formulas for data-driven control: Stabilization,
  optimality, and robustness,'' \emph{IEEE Transactions on Automatic Control},
  vol.~65, no.~3, pp. 909--924, 2019.

\bibitem{mania2019certainty}
H.~Mania, S.~Tu, and B.~Recht, ``Certainty equivalence is efficient for linear
  quadratic control,'' \emph{Advances in Neural Information Processing
  Systems}, vol.~32, 2019.

\bibitem{tsiamis2021linear}
A.~Tsiamis and G.~J. Pappas, ``Linear systems can be hard to learn,'' in
  \emph{2021 60th IEEE Conference on Decision and Control (CDC)}.\hskip 1em
  plus 0.5em minus 0.4em\relax IEEE, 2021, pp. 2903--2910.

\bibitem{zhou1998essentials}
K.~Zhou and J.~C. Doyle, \emph{Essentials of robust control}.\hskip 1em plus
  0.5em minus 0.4em\relax Prentice hall Upper Saddle River, NJ, 1998, vol. 104.

\bibitem{georgiou1990optimal}
T.~Georgiou and M.~Smith, ``Optimal robustness in the gap metric,'' \emph{IEEE
  Transactions on Automatic Control}, vol.~35, no.~6, pp. 673--686, 1990.

\bibitem{jedra2019sample}
Y.~Jedra and A.~Proutiere, ``Sample complexity lower bounds for linear system
  identification,'' in \emph{2019 IEEE 58th Conference on Decision and Control
  (CDC)}.\hskip 1em plus 0.5em minus 0.4em\relax IEEE, 2019, pp. 2676--2681.

\bibitem{li2022fundamental}
J.~Li, S.~Sun, and Y.~Mo, ``Fundamental limit on siso system identification,''
  in \emph{2022 IEEE 61st Conference on Decision and Control (CDC)}.\hskip 1em
  plus 0.5em minus 0.4em\relax IEEE, 2022, pp. 856--861.

\bibitem{ziemann2022policy}
I.~Ziemann, A.~Tsiamis, H.~Sandberg, and N.~Matni, ``How are policy gradient
  methods affected by the limits of control?'' in \emph{2022 IEEE 61st
  Conference on Decision and Control (CDC)}.\hskip 1em plus 0.5em minus
  0.4em\relax IEEE, 2022, pp. 5992--5999.

\bibitem{ackermann1972entwurf}
J.~Ackermann, ``Der entwurf linearer regelungssysteme im zustandsraum,''
  \emph{at-Automatisierungstechnik}, vol.~20, no. 1-12, pp. 297--300, 1972.

\bibitem{fadali2012digital}
M.~S. Fadali and A.~Visioli, \emph{Digital control engineering: analysis and
  design}.\hskip 1em plus 0.5em minus 0.4em\relax Academic Press, 2012.

\bibitem{boucheron2013concentration}
S.~Boucheron, G.~Lugosi, and P.~Massart, \emph{Concentration inequalities: A
  nonasymptotic theory of independence}.\hskip 1em plus 0.5em minus 0.4em\relax
  Oxford university press, 2013.

\end{thebibliography}
\end{document}